\documentclass[%
reprint,
eprint,
longbibliography
]{revtex4-2}

\usepackage{graphicx}
\usepackage{dcolumn}
\usepackage{bm}
\usepackage{amsmath,amssymb,amsthm,easybmat,verbatim}
\usepackage{enumerate}
\usepackage{hyperref}
\hypersetup{colorlinks=true, linkcolor=blue!80!black, urlcolor=blue!80!black, citecolor=blue!80!black}
\usepackage{color}
\usepackage{longtable}
\usepackage{array}
\usepackage{tikz}
\usepackage{float}
\usepackage[normalem]{ulem}

\newtheorem{corollary}{Corollary}
\newtheorem{conjecture}{Conjecture}

\newtheorem{theorem}{Theorem}

\usepackage{physics}
\usepackage{mathtools}

\newcommand{\defined}{\triangleq}

\begin{document}

\title{An Interacting System and Environment with Prior Correlations Plus Local Operations Can Mimic Uncorrelated Evolution}

\author{Daniel Dilley$^{3}$}
\email{ddilley@anl.gov}
\author{Alvin Gonzales$^{3}$}
\email{agonza@siu.edu}
\author{Jeffrey Larson$^3$}
\email{jmlarson@anl.gov}
\author{Mark Byrd$^{1,2,4}$}
\email{byrdms@ornl.gov}
\affiliation{$^1$School of Computing, Southern Illinois University Carbondale}
\affiliation{$^2$School of Physics and Applied Physics, Southern Illinois University Carbondale.}
\date{\today}
\affiliation{$^3$Mathematics and Computer Science Division, Argonne National Laboratory}
\affiliation{$^4$ (Present Address) Computational Sciences and Engineering Division, Oak Ridge National Laboratory}

\begin{abstract}
Initial system--environment correlations can induce reduced dynamics that depart from the standard completely positive (CP) description. We study when such dynamics can be reproduced by the same global unitary acting on an initial system--environment product state. We show that local preprocessing on the system, applied prior to the joint evolution, can lead to dynamics which are completely positive.  We focus on two-qubit system--environment states, with a three-qubit extension when a single ancilla is used to implement Kraus channel preprocessing. We analyze three classes of local operations---measurements, unitaries, and Kraus channels---and characterize their ability to satisfy dynamics matching while minimally perturbing the initial system state. Local measurements can always enforce dynamics matching but necessarily reduce state fidelity. We numerically investigated local unitary preprocessing. Local unitaries enforce dynamics matching in all tested instances and typically preserve the system state: in 402 numerical cases the average fidelity is $99.8\%$, more than $92\%$ of cases exceed $99\%$ fidelity, and the minimum fidelity is $94.2\%$. We show that a two-term Kraus channel, realizable with a single ancillary qubit coupled to the system, achieves dynamics matching with unit fidelity in all tested cases.

As a second contribution, we demonstrate that local preprocessing can prevent the emergence of non--completely positive (NCP) reduced dynamics in a time-dependent setting. For a correlated two-qubit experiment with a time-dependent global unitary $U(t)$, the induced compatibility domain is fixed by the initial correlations. For every $t$ in a nontrivial continuous interval, the reduced dynamics are NCP when defined over this domain. We then show that a single, time-independent local unitary applied prior to the joint evolution renders the reduced dynamics CP over the entire interval.
\end{abstract}

\maketitle

\section{Introduction} The dynamics of an open quantum system are shaped by interactions with the environment and any initial correlations present between them. A common assumption in open quantum systems theory is that the system and environment are initially uncorrelated, in which case the reduced system dynamics are described by a completely positive (CP) map \cite{Nielsen_Chuang_Textbook_2011, breuer2002theory,rivas2012open}.  
In realistic settings, however, state preparation is imperfect, and residual system--environment correlations are often unavoidable \cite{modi_2012_Positivity,Brodutch_Datta_Modi_Rivas_Rosario_2013,liu_2014,rodriguez-rosario_2008,Dominy_Shabani_Lidar_2016,Dominy_Lidar_2016}.   Correlated initial states can then give rise to reduced dynamics that are not completely positive, complicating interpretation and limiting the applicability of product-state models \cite{Nielsen_Chuang_Textbook_2011, Pechukas_1994, Pechukas_1995}.

This work addresses two operational questions about mitigating the effects of initial correlations using interventions on the system alone. The first is a \textit{dynamics-matching} problem: Given reduced dynamics arising from a correlated initial state, what local operations to the system can we apply, before the joint evolution, such that the resulting reduced dynamics are exactly those that would be obtained from an initial product system--environment state under the same global unitary? The second concerns complete positivity: In time-dependent scenarios where reduced dynamics are naturally defined only on a compatibility domain fixed by the initial correlations \cite{Shaji_Sudarshan_2005}, what local preprocessing prevents the appearance of non-completely positive (NCP) reduced dynamics? We study these questions in a two-qubit setting, where the system dynamics are obtained by tracing out one qubit after a global unitary evolution.  

Within this framework we consider three families of local preprocessing operations---measurements, unitaries, and general Kraus channels---and quantify how well each enforces dynamics matching while minimally disturbing the system state. To compare strategies on equal footing, we optimize the fidelity between the system state before and after preprocessing, thereby identifying the least invasive local intervention that restores product-state-consistent dynamics. 

Our numerical results show a clear set of trade-offs. Local measurements can always enforce dynamics matching, but they do so by disturbing the system state and thus lowering fidelity. Local unitaries require no ancillary systems and enforce dynamics matching in all 402 tested instances; across this dataset they achieve an average fidelity of approximately $99.8\%$, exceed $99\%$ fidelity in more than $92\%$ of cases, and attain a minimum fidelity of $94.2\%$. When ancillary resources are available, we find that a two-term Kraus channel---implementable with a single ancilla qubit---enforces dynamics matching with unit fidelity in all 402 cases studied, providing a non-disturbing solution in this setting.

Next, we present a proof-of-principle calculation of NCP-to-CP mitigation via local preprocessing. For a correlated two-qubit experiment with a time-dependent joint unitary $U(t)$, we show that the resulting reduced dynamics are NCP for every $t$ in a continuous time interval when defined over the induced compatibility domain. We then show that a single, time-independent local unitary applied prior to the global evolution suffices to render the reduced dynamics completely positive over the entire interval. Together, these findings clarify the operational role of local preprocessing in open quantum systems and delineate the trade-offs between fidelity and resource requirements. All numerical data and code supporting our results are publicly available \cite{github_dilley}.

\section{Background}
We consider a quantum system $S$ interacting with an environment $E$ via a global unitary evolution $U$. We do not assume that the initial joint state $\rho^{SE}$ is uncorrelated. The reduced system state before and after the global evolution is
\begin{align}
\rho^S \defined \tr_E(\rho^{SE}), \qquad \rho_U^S \defined \tr_E\!\left(U\,\rho^{SE}\,U^\dagger\right).
\end{align}
We allow for any kind of local preprocessing channel prior to global evolution. Let
\begin{align}
\rho^{S}_{\Lambda} \defined \tr_E(\Lambda_S(\rho^{SE}))
\end{align}
where $\Lambda_S(\star)$ is the preprocessing channel on the system $S$. The channel reduces to identity when no preprocessing is applied. A central question in this work is when the observed input-output transformation can be reproduced as if it originated from an initial product system--environment state. We say that the system dynamics are \textit{$U$-generated by a product state} if there exists a valid environment state $\zeta^E$ such that \begin{align}\label{eq:initialCond0} \tr_E\left(U\,\Lambda_S(\rho^{SE})\,U^\dagger\right) = \tr_E\left(U\,(\Lambda_S(\rho^S) \otimes \zeta^E)\,U^\dagger\right). \end{align} We refer to Eq.~\eqref{eq:initialCond0} as the \textit{dynamics-matching condition}. It concerns a single system input state and does not assume that the reduced dynamics is a completely positive map.

Our goal is to enforce the dynamics-matching condition using local preprocessing operations acting only on the system before the joint evolution. Throughout, we compare different preprocessing strategies by optimizing the fidelity between the system state before and after the local preprocessing,
\begin{align}
    F(\rho^S, \rho^S_\Lambda),
\end{align}
thereby identifying the least invasive intervention that still enforces dynamics matching. 

We first consider local measurements as a route to enforcing Eq.~\eqref{eq:initialCond0}. Following Brun's generalized measurement framework~\cite{Brun_2002}, we model a two-outcome measurement using Hermitian operators
\begin{align}
\label{Eq:measurement}
M_{\pm}(\epsilon,\hat{n}) \defined \dfrac{1}{2} \big( \varepsilon_+ \mathbb{I} \pm \varepsilon_- \hat{n}\cdot\vec{\sigma} \big),
\end{align}
where
\begin{align}
\label{Eq:Measurement_Coefficients} \varepsilon_{\pm} \defined \sqrt{\tfrac{1+\epsilon}{2}} \pm \sqrt{\tfrac{1-\epsilon}{2}}, \qquad 0 \leq \epsilon \leq 1,
\end{align}
and $\hat{n}$ is a unit vector. These operators satisfy the completeness relation $M_+^2 + M_-^2 = \mathbb{I}$. The parameter $\epsilon$ controls the measurement strength, interpolating between weak ($\epsilon \approx 0$) and projective ($\epsilon = 1$) measurements. We focus on weak measurements because they disturb the initial state less strongly. Measurements provide a guaranteed mechanism for enforcing dynamics matching but generally do so at the cost of disturbing the system state.

For concreteness, we specialize throughout in a two-qubit system--environment model, with joint initial state
\begin{align}
\rho^{SE} = \dfrac{1}{4}\left( \mathbb{I}\otimes\mathbb{I} + \vec{a}\cdot\vec{\sigma}\otimes\mathbb{I} + \mathbb{I}\otimes\vec{b}\cdot\vec{\sigma} + \sum_{i,j=1}^3 t_{ij}\, \sigma_i \otimes \sigma_j \right),
\end{align}
where $\vec{a}$ and $\vec{b}$ are the Bloch vectors of the system and environment, respectively, and $t_{ij}$ is the correlation matrix. The joint dynamics are governed by a global unitary $U$.

\section{Results}
In this section we present three strategies for ensuring that the system dynamics can be $U$-generated from an initial product state. We begin with the case where the original pair $(U,\rho^{SE})$---that is, the preprocessing channel is identity---does not satisfy the dynamics-matching condition. However, the condition can be enforced by allowing a preprocessing local operation $G$ on the system. Accordingly, we redefine the experiment as the tuple $(U,\,(G \otimes \mathbb{I})\,\rho^{SE}\,(G \otimes \mathbb{I})^\dagger).$ This construction modifies system--environment correlations without requiring direct control of the environment. In general, the choice of $G$ depends on both the global unitary $U$ and the initial joint state $\rho^{SE}$. While this limits the generality of closed-form statements, it still allows us to identify practically useful strategies and to quantify the benefits of modest additional resources, such as a single ancilla qubit, for implementing more general local channels.

The first strategy takes $G$ to be a local measurement operation. This discussion is in Appendix \ref{sec:TransfomToCorr}-\ref{sec:Werner}. In Appendix~\ref{Sec:Repeated_Measurements} we show that a sequence of repeated measurements in Eq.~\eqref{Eq:measurement} is generally not equivalent to a single measurement with a different choice of $\epsilon$ when the direction is fixed. In particular, repeated preprocessing can induce system dynamics that cannot, in general, be reproduced by any single measurement of the form in Eq.~\eqref{Eq:measurement}. We use measurements as a motivating example because, although they introduce irreversibility and can disrupt the system state, they can always be tuned to perturb correlations enough to enforce product-state $U$-generation. In the extreme case, a projective measurement destroys all system--environment correlations, yielding a product state by construction. The trade-off is a potentially substantial loss of system-state fidelity. This feature may nevertheless be useful in settings such as stabilizer measurements, where the data qubits are not measured directly. In such cases, projective measurements on an ancillary subsystem may help suppress detrimental correlated errors.

We next examine local unitary preprocessing. Local unitaries preserve the system spectrum, require no ancillary resources, and can often satisfy the dynamics-matching condition while only weakly perturbing the system state. We numerically examine their performance over 402 random two-qubit initial global states. We were able to find a local unitary that enforces dynamics matching in all 402 tested instances. Across this dataset we achieve an average fidelity of approximately $99.8\%$, exceed $99\%$ fidelity in more than $92\%$ of cases, and attain a minimum fidelity of $94.2\%$.
\textbf{GAIL - Note that the preceding sentence repeats  that in the Intro. Perhaps the details should be deleted from there.}

Let $\zeta^E = \sum_j p_j\,\ket{\zeta_j}\!\bra{\zeta_j}$ and $\{\ket{\eta}^E\}$ be an orthonormal basis  for the environment in Eq.~\eqref{eq:initialCond0}. Then 
\begin{align}
\label{Eq:Kraus_Solution} \tr_E\left(U\,(\rho^S \otimes \zeta^E)\,U^\dagger\right) = \sum_{j,\eta} p_j\, K_{j\eta}\,\rho^S\,K_{j\eta}^\dagger,
\end{align}
with Kraus operators $K_{j\eta} \defined {}^E\!\bra{\eta}\,U\,\ket{\zeta_j}^E$. This does not imply that the reduced dynamics define a CPTP map on the compatibility domain \cite{Shaji_Sudarshan_2005}, but this observation motivates our study of local Kraus channels as preprocessing operations. In particular, we consider two-term Kraus channels implemented by  using a single ancillary qubit. These channels constitute a minimal non-unitary extension beyond local unitaries and provide, in all tested cases, a mechanism for enforcing the dynamics-matching condition while preserving the system state with unit fidelity.

We recall that when the system and environment are initially correlated, the reduced dynamics need not be completely positive~\cite{Pechukas_1994,Pechukas_1995,Alicki_1995,Jordan_Shaji_Sudarshan_2004}. Such NCP behavior reflects the restricted set of physically admissible system inputs---that is, the compatibility domain induced by the initial correlations---rather than a failure of quantum mechanics. In later sections we give a proof-of-principle calculation that local unitary preprocessing can eliminate NCP reduced dynamics for a correlated two-qubit experiment with a time-dependent global evolution. This result is conceptually distinct from the fidelity optimization used to compare dynamics-matching strategies.


\subsection{Application of Local Unitaries}
We consider the example given in Appendix
\ref{Sec:Full_Projection_Example} but replace local measurements with arbitrary local unitaries applied to the system prior to the global evolution. Once again we take $U = SWAP\cdot CNOT$, but now we rotate the initial joint state by a system-only unitary,
\begin{align} \label{Eq:Exp_with_theta}
\rho^{SE}(\theta) = (R_y(\theta) \otimes \mathbb{I}) \Phi^+ (R_y(\theta)^\dagger \otimes \mathbb{I}),
\end{align}
where $R_y(\theta) = e^{-i \theta/2 \hat{\sigma}_y}$. Because the reduced system state is maximally mixed, any local unitary preserves the system state exactly. This is in sharp contrast to local measurements in the same scenario, which necessarily disturb the system.

We apply the measurement $M_+$ in Eq.~\eqref{Eq:measurement} to the state in Eq.~\eqref{Eq:Exp_with_theta} prior to global evolution $U$. Figure \ref{Fig:Swap_Cnot_Plot} reports the minimum $\epsilon$ of our measurement required for the system dynamics to be $U$-generated by a product state as $\theta$ varies over $[0,\pi/2]$. This data illustrates the expected fidelity trade-off: as the required measurement strength increases, the optimal fidelity decreases. In particular, measurements cannot preserve fidelity perfectly unless no measurement is required. By contrast, local unitaries can steer the joint correlations while leaving the state of the system unchanged. For example, at $\theta = \pi/2$, no measurement is necessary, and the system dynamics can be $U$-generated from a product state. This shows that local basis changes can be sufficient to restore product-state $U$-generation while preserving the initial state of the system exactly.
\begin{figure}[h]
\centering
\includegraphics[scale = 0.40]{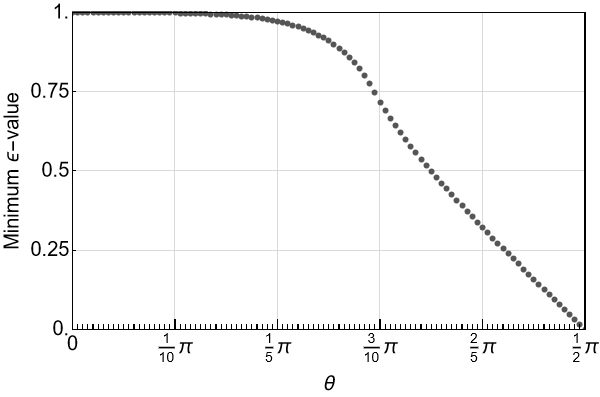}
\caption{Plot of the minimum $\epsilon$-value needed to $U$-generate the system dynamics, $\rho^S (\theta) \xrightarrow{U} \tr_E (U \rho^{SE}(\theta) U^\dagger)$, using a product state. The rotation angle $\theta$ for the local unitary is varied between 0 and $\pi/2$.} \label{Fig:Swap_Cnot_Plot}
\end{figure}
This example motivates a broader question: Can a local unitary on the system always be chosen so that the resulting reduced dynamics are $U$-generated by an initial product state?

As an illustration, for the Werner state $W(\lambda)$ in Appendix \ref{sec:Werner} and global evolution given by a CNOT gate, the local unitary $e^{-i \pi/4 \hat{\sigma}_y}$ suffices for any $\lambda$. 
Operationally, this approach avoids the destructiveness of measurements by replacing correlation erasure with a controlled, orthogonal transformation of the correlation matrix. Correlations are not eliminated; rather, they are reshaped in a predictable way that can restore dynamics matching.

\subsection{Classes of Global Unitaries Admitting Product-State Dynamics Matching} We now analyze how the structure of the global unitary evolution influences the feasibility of enforcing the dynamics-matching condition via local operations on the system. Specifically, we classify two-qubit global unitaries according to their nonlocal parameters and determine when local preprocessing can always be used to $U$-generate the observed system dynamics from a product state. Any two-qubit unitary $U$ admits a Cartan $KAK$ decomposition~\cite{Helgason_1978, Knapp_2002, Tucci_2005}, separating $U$ into local rotations and a nonlocal part parameterized by three real numbers $\alpha_1, \alpha_2, \alpha_3$. The nonlocal component is given by
\begin{align}
\label{Eq:Nonlocal_Part}
\Omega \defined e^{-i (\alpha_1 \hat{\sigma}_1 \otimes \hat{\sigma}_1 + \alpha_2 \hat{\sigma}_2 \otimes \hat{\sigma}_2 + \alpha_3 \hat{\sigma}_3 \otimes \hat{\sigma}_3)}.
\end{align}
We define the set $\mathcal{S} = \{ n\pi/2 : n \in \mathbb{W} \}$ and classify unitaries according to whether their nonlocal parameters $\alpha_i$ belong to $\mathcal{S}$. In particular, we distinguish two-parameter families, where only one $\alpha_i$ is in $\mathcal{S}$, from more general three-parameter cases. To facilitate the analysis, we introduce the Givens rotation matrix $g_{ab}$, which acts on the correlation components of a two-qubit state and can always be implemented by local unitaries on the system. The nonzero elements of $g_{ab}$ are defined as follows:
\begin{itemize} \item $g_{aa} = 1$ for $a \neq i, j$,
\item $g_{aa} = \cos(\theta)$ for $a = i, j$,
\item $g_{ji} = -g_{ij} = -\sin(\theta)$ for fixed $i > j$,
\end{itemize}
where this rotation applies a counterclockwise rotation in the $(i, j)$ plane by angle $\theta$. We now state our first main result for this section.
\begin{theorem}[Local Unitary]
\label{Thm:Two_Parameter_Theorem}
Let $\Omega$ be the nonlocal part of the global evolution as defined above. If $U$ is part of a two-parameter family---that is, only one of the parameters $\alpha_1, \alpha_2, \alpha_3$ is in $\mathcal{S}$---then suitable Givens rotations on the system can always be used to ensure that the resulting system dynamics are $U$-generated by an initial product state. \end{theorem}
\noindent The proof of Theorem~\ref{Thm:Two_Parameter_Theorem} is provided in Appendix~\ref{Supp:Theorem_1}.

To determine the explicit local unitaries required to implement these Givens rotations, one may use the formula from \cite{Dilley_2022} to modify the correlation matrix. For example, if $0 < x < \pi/2$, the first local unitary is given by
\begin{align}
U_1 \defined \frac{1}{2}
\begin{pmatrix}
a - i b & -a - i b \\
a - i b & a + i b,
\end{pmatrix}
\end{align}
where $a \defined \sqrt{1-\sin(x)}$ and $b \defined \sqrt{1+\sin(x)}$. This unitary induces the first Givens rotation on the left side of the correlation matrix of $\rho^{SE}$. In Appendix~\ref{Sec:Numerical_Local_Operations} we numerically investigate the more general case of three-parameter unitaries, where $\alpha_i \notin \mathcal{S}$ for all $i \in \{1, 2, 3\}$.

Beyond the general case, it is instructive to consider the special scenario in which the initial system--environment state possesses a diagonal correlation matrix. This class of states is particularly relevant because, up to local unitaries, any two-qubit state can be transformed into a form with diagonal correlations. Consequently, the analysis of diagonal states encompasses a wide range of physically meaningful situations and provides insight into the universality of local control strategies.

Formally, let $\rho^{SE}_d$ denote a two-qubit state whose correlation matrix is diagonal, and let the global evolution be given by $\Omega$ as defined in Eq.~\eqref{Eq:Nonlocal_Part}. Although this restriction may appear specialized, it captures every instance equivalent up to local unitaries; that is, any tuple of the form $\{ (\mathcal{L}_1 \otimes \mathcal{L}_2) \Omega (\mathcal{R}_1 \otimes \mathcal{R}_2), (\mathcal{R}_1 \otimes \mathcal{R}_2)^\dagger \rho^{SE}_d (\mathcal{R}_1 \otimes \mathcal{R}_2) \}$ will exhibit the same properties for the purposes of dynamics matching. In effect, the right local unitaries of the global evolution $U$ are matched to the conjugate transpose of the local unitaries acting on the state $\rho^{SE}_d$.

This observation leads to the following theorem, which establishes that local unitary preprocessing is always sufficient to enforce the dynamics-matching condition for diagonally correlated states.
\begin{theorem} \label{Thm:Diagonal_States}
For all nonlocal unitaries $\Omega$ given in Eq.~\eqref{Eq:Nonlocal_Part} and two-qubit states $\rho^{SE}_d$ whose correlation matrix is diagonal, the system dynamics can always be $U$-generated by a product state.
\end{theorem}
\noindent The proof of Theorem~\ref{Thm:Diagonal_States} is provided in Appendix~\ref{Supp:Theorem_2}. This result further highlights the operational power of local unitaries: whenever the initial correlations are diagonal---or can be diagonalized by local rotations---the observed system dynamics can be reproduced by the same global unitary acting on a product state, without the need for preprocessing.

In addition to the case of diagonally correlated states, another important scenario arises when the global unitary evolution is incapable of generating entanglement from any input state. This restriction, while seemingly strong, has significant implications for the structure of the reduced dynamics. Even though separable states can sometimes produce non-completely positive dynamics that cannot be reproduced by using a product state, the absence of entangling power in the global evolution ensures that, in the present two-qubit setting, the resulting system dynamics always admit a product-state description. 

We formalize this observation in the following theorem.
\begin{theorem}[Non-entangling Unitaries]
\label{Thm:Non-Entangling_Unitary}
If for any two-qubit state $\rho^{SE}$, the nonlocal part of the unitary in Eq.~(\ref{Eq:Nonlocal_Part}) satisfies the condition
\begin{align} \label{Eq:Non_Entangling_Condition} \frac{3 - c_1(c_2 + c_3) - c_2 c_3}{4} = 0,
\end{align}
where $c_i \defined \cos(4 \alpha_i)$, then the dynamics on the system can be $U$-generated by a product state.
\end{theorem}
\noindent The proof of Theorem~\ref{Thm:Non-Entangling_Unitary} can be found in Appendix~\ref{Supp:Theorem_3}.

A natural question is whether these results can be extended to systems beyond two qubits, or even to general multipartite unitaries. Prior work on the entangling power of broader classes of unitary evolutions can be found in \cite{Scott2004, Linowski2020}, but a systematic extension of our framework to these settings remains open for future research.

Another intriguing direction concerns the role of local unitaries in shaping the system--environment correlations. Specifically, consider the transformation of the joint state under
\begin{align}
U (V \otimes \mathbb{I}) \rho^{SE} (V^\dagger \otimes \mathbb{I}) U^\dagger,
\end{align}
where $V$ is a local unitary acting only on the system. If one could always choose a local unitary $V$ on the system such that the entanglement of $\rho^{SE}$ is unchanged by the subsequent action of $U$, then the resulting reduced dynamics would always admit a product-state description. Indeed, a two-qubit unitary with zero entangling and disentangling power is locally equivalent either to the identity or to the SWAP operation. In that case, after the preprocessing $V \otimes \mathbb{I}$, the joint evolution $U$ would leave the state within the same local-unitary equivalence class or map it to one related by a local unitary SWAP equivalent. Since both situations are already known to satisfy the dynamics-matching condition, the system dynamics could then always be $U$-generated by an initial product state. This perspective suggests a more general connection between local control of correlations and the ability to recover a product-state description of the reduced dynamics.

\subsection{Maintaining Fidelity by Using Kraus Channels} The numerical analysis in Appendix~\ref{Sec:Numerical_Local_Operations} revealed that local unitary preprocessing can enforce the dynamics-matching condition for a wide range of two-qubit experiments, often with very high fidelity. In a subset of challenging cases, however, perfect fidelity could not be achieved by using unitary operations alone. Specifically, out of $1{,}000$ constructed experiments $(U_i, \rho^{SE}_i)$ designed to make Eq.~(\ref{eq:initialCond0}) difficult to satisfy, $402$ nontrivial instances were identified in which the inferred product-state environment $\zeta^E$ would otherwise be nonphysical (i.e., $|\vec{\zeta}| > 1$). For these cases, a single local unitary on the system $S$ was always sufficient to enforce dynamics matching, but the fidelity between the pre- and postprocessed initial system states did not always reach unity (see Fig.~\ref{Fig:Fidelity_vs_State}).

To address this inherent limitation of deterministic unitary preprocessing, we investigated whether more general local quantum channels---specifically, Kraus channels---could enforce Eq.~(\ref{eq:initialCond0}) while preserving the system state perfectly. Our goal was to identify the minimal non-unitary extension, ideally implementable with only a single ancillary qubit, that would achieve unit fidelity without altering the global evolution $U$.

We first considered local Kraus channels of the form
\begin{align}
\mathcal{E}(\rho) = \left( 1 - \sum_{j = 1}^3 p_j \right) \rho + \sum_{j = 1}^3 p_j (V_j \otimes \mathbb{I}) \rho (V_j^\dagger \otimes \mathbb{I}),
\end{align}
where the $V_j$ are local unitaries that rotate about the Bloch vector of the system and the $p_j$ are probabilities. These channels are appealing because they can be implemented without ancillary qubits, relying only on probabilistic mixtures of local unitaries. However, even in the worst-case scenarios---where the optimal unitary yields fidelities between $94\%$ and $95\%$---these channels increase the fidelity only marginally, to approximately $96\%$. 

Motivated by these results, we next examined more general fidelity-preserving Kraus maps. We found that quantum channels of the form
\begin{align}
\label{Eq:Kraus_Channel}
\mathcal{E}(\rho) = \tfrac{1}{2} \sum_{j = 1}^2 (V_j \otimes \mathbb{I}) \rho (V_j^\dagger \otimes \mathbb{I}), \end{align}
where $V_1$ and $V_2$ are local operations and the channel satisfies the completeness relation, allow us to $U$-generate the system dynamics using a product state with unit fidelity for all $402$ challenging cases. Restricting to two Kraus operators ensures that the channel can be implemented by using a single ancillary qubit via the Stinespring dilation~\cite{Stinespring1955}.

The resulting system dynamics can be written as 
\begin{align}
\notag\tr_{AE} \left[ (\mathbb{I} \otimes U) (W \otimes \mathbb{I}) \left( \ket{0}\bra{0}^A \otimes \rho^{SE} \right)\right.\\
\left. (W^\dagger \otimes \mathbb{I}) (\mathbb{I} \otimes U^\dagger) \right],
\end{align}
where $W$ is a two-qubit unitary acting on the ancillary qubit $A$ and the system $S$, used to induce the Kraus channel in Eq.~(\ref{Eq:Kraus_Channel}) on $S$. The unitary $W$ can be constructed directly from the Kraus operators by using the Gram--Schmidt process~\cite{GolubVanLoan1996}. In most cases, the channel reduces to a simple rotation about the local Bloch vector of $S$.

\subsection{Preventing NCP Dynamics} \label{Sec:NCP_Map_Example} We now present an example demonstrating how a local unitary operation can be used to prevent non-completely positive  dynamics on the system. To do so, we first define the reduced dynamical map in the context of a time-dependent global evolution $U(t)$ and the associated compatibility domain of system states.

Rather than considering the static experiment $\{\mathrm{CNOT}, \Phi^+\}$, we focus on the time-dependent experiment $\{e^{-i t\, \mathrm{CNOT}}, \Lambda^{SE}_p\}$, where the initial correlated state is given by
\begin{align}
\Lambda_p^{SE} (\tau^S) \defined \tau^S \otimes \frac{\mathbb{I}}{2} + (1-p) \left( \Phi^+ - \frac{\mathbb{I}}{2} \otimes \frac{\mathbb{I}}{2} \right),
\end{align}
with $\tr_E [\Lambda_p^{SE}(\tau^S)] = \tau^S$. Here, $U=e^{-it \text{CNOT}}$. Note that $\Lambda_p^{SE}$ is the assignment map \cite{Pechukas_1994, Shaji_Sudarshan_2005}.

The compatibility domain consists of valid input states $\tau^S$ whose Bloch vectors satisfy $\vec{\tau} = p \vec{\mu}$, where $\vec{\mu}$ is real and $\|\vec{\mu}\| \leq 1$. For fixed $p$, both the system--environment correlations and the local environmental state are determined. When $p = 1$, the system dynamics can be $U$-generated by a product state; but for $p < 1$, the dynamical map must be analyzed for all possible inputs.

The dynamical matrix that describes the subsystem dynamics for this experiment is
\begin{align} \label{Eq:A_Matrix}
A_S =
\begin{pmatrix}
1 & 0 & 0 & 0 \\
a & e^{-it}\cos(t) & 0 & a \\
a^* & 0 & e^{it}\cos(t) & a^* \\
0 & 0 & 0 & 1
\end{pmatrix},
\end{align}
where $a \defined (1-e^{-2it})(1-p)/4$ \cite{sudarshan_1961}. This matrix is obtained by vectorizing the state $\tr_E [\Lambda_p^{SE}(\tau^S)]$ and applying the general single-qubit trace-preserving map as described in Eq.~59 of~\cite{Dilley_2021}. The parameters of $A_S$ are chosen so that the dynamics yield the output state $\tr_E [U \Lambda_p^{SE}(\tau^S) U^\dagger]$ for all admissible inputs $\tau^S$. Equation~\eqref{Eq:A_Matrix} is the only valid $A_S$-matrix that does this.

The associated realigned $B_S$-matrix~\cite{Devi_2011} has the negative eigenvalue
\begin{align}
\lambda = \frac{1}{2} \left( 1 - \cos(t) - f(p,t) \sin\left(\frac{t}{2}\right) \right),
\end{align}
for $0 < p < 1$ and all times $t \in (0, \pi/8]$, where
\begin{align}
f(p,t) = \sqrt{2 \left[ 2 + (-2 + p)p + (-2 + p)p \cos(t) \right]},
\end{align}
which indicates that the dynamics on system $S$ are non-completely positive immediately after the global evolution begins, even when the state is separable. For $t = 0$, the map is the identity; and for $p = 1$, the dynamics are completely positive for all $t$.

We note that, although the experiment at $p = 1$ is described by a CP map for all $t$, the system dynamics cannot always be $U$-generated by a product state, particularly at $t = \pi/2$. 
In this example, NCP behavior arises only in the nonstationary setting, where the global unitary is time dependent and the reduced map is defined on a restricted compatibility domain.

Remarkably, a single local unitary can be used to prevent all NCP dynamics in this setting. For example, applying the local rotation $R_Y(\pi/2) \defined e^{-i (\pi/4) \sigma_y}$ to the system prior to the global evolution transforms the subsystem dynamics as follows:
\begin{align}
& R_Y(\pi/2) \tr_E (\Lambda^{SE}_p) R_Y^T(\pi/2) \\
\rightarrow & \tr_E \left[ U \left( R_Y(\pi/2) \otimes \mathbb{I} \right) \Lambda^{SE}_p \left( R_Y^T(\pi/2) \otimes \mathbb{I} \right) U^\dagger \right],
\end{align}
resulting in the new dynamical matrix
\begin{align}
A_S =
\begin{pmatrix}
1 & 0 & 0 & 0 \\
0 & e^{-it}\cos(t) & 0 & 0 \\
0 & 0 & e^{it}\cos(t) & 0 \\
0 & 0 & 0 & 1
\end{pmatrix},
\end{align}
which always yields a realigned $B_S$-matrix with non-negative eigenvalues.

Therefore, a single local rotation on the system suffices to prevent all NCP dynamics throughout the experiment. This example illustrates that a simple rotation of the initial system state can profoundly affect the nature of the subsystem dynamics induced by the environment and global evolution.  

\section{Conclusion}

We addressed the operational question of when reduced system dynamics arising from an initially correlated system--environment state $\rho^{SE}$, under a fixed global unitary $U$, can be reproduced as 
\begin{align}
\operatorname{Tr}_E\!\left[ U \bigl( \rho_S \otimes \zeta_E \bigr) U^\dagger \right].
\end{align}
 Our focus was on constructive, experimentally relevant ways to enforce this dynamics matching condition while quantitatively minimizing disturbance to the system. Our contributions are threefold. First, we introduced a measurement-based prescription that can always modify system--environment correlations so that the dynamics-matching condition is satisfied. This method provides an analytic route to construct the required environmental state $\zeta^E$ and to compute minimal measurement strength in representative families (e.g., Werner and Bell-derived examples), but it is inherently invasive and typically lowers the fidelity between the pre- and postpreparation system states.

Second, we showed that local unitary preprocessing is a far less disruptive alternative and is often sufficient. For any two-qubit correlated state and any two-parameter family of global unitaries, we proved that an appropriate local unitary on the system always suffices to $U$-generate the observed dynamics from a product state. For more general three-parameter nonlocal two-qubit unitaries, our numerical searches consistently support the same conclusion; in most cases, the optimal operation is a rotation about the system’s Bloch vector, preserving the system state exactly. These results yield a theoretical pathway: Given $\{U, \rho^{SE}\}$, one (i) searches for a Bloch-vector rotation that enforces dynamics matching and, when needed, (ii) augments this with a two-term Kraus implementation using a single ancilla qubit to reach unit fidelity in cases where unitaries alone cannot. In all tested cases, including extreme instances, this strategy succeeded in achieving perfect-fidelity product-state $U$-generation.

Third, we provided a proof-of-principle calculation for preventing NCP reduced dynamics using a single local rotation in a time-dependent setting. For the family of correlated states we defined and for all times $t$ under the chosen global evolution, preparing the system in the rotated basis $\{\ket{\tilde{+}}, \ket{\tilde{-}}\}$ eliminates NCP behavior; operationally, this corresponds to a rotation about the $y$-axis and can remove NCP dynamics without further intervention. This highlights that local basis choice can qualitatively change the reduced dynamics induced by fixed correlations and interactions, and it motivates deeper study of how local unitaries steer open-system behavior under time-dependent evolution.

We also clarified practical limitations and information requirements. Our constructive procedures assume access to the joint correlations $\mathcal{T}_\rho$, the local Bloch vectors of each qubit, and the global evolution $U$, which may be difficult to estimate. In many experimental systems, however, the primary system-environment interaction is known.   The interaction Hamiltonian is also a central issue in dynamical decoupling (DD) \cite{Viola1999_DynamicalDecoupling,Lidar/Brun:qec}, and such approaches in DD may be relevant to this situation. Reducing these assumptions is a central direction for future work.  The work here is also important for identifying correlations between the system and its environment.  The ability for a system-environment evolution to be describable by an initial product state can indicate the degree of robustness of methods for identifying initial correlations \cite{Kimura:07,chitambar2015,Modi_2012,RingbauerEtAl:15,Hagen:21}.  Notably, in the NCP-mitigation example, we did not require knowledge of the specific system input, but we did require knowledge of the system--environment interaction, which may be available in architectures where the noise model (and hence the interaction Hamiltonian) is characterized or empirically known.  

Several concrete  steps follow. One is to incorporate more realistic partial-information models, for example by bounding correlation strength via a $\bigl\| T_\rho - \vec{a}\,\vec{b}^{\,T} \bigr\| \le \kappa$, and to determine whether sufficiently small $\kappa$ guarantees product-state $U$-generation using a broader set of local unitaries. Similarly, because mitigating NCP dynamics appears to require less information than enforcing dynamics matching, an important open question is what minimal knowledge is needed to suppress or circumvent non-Markovian or otherwise detrimental behavior and what additional pathologies can be avoided via simple basis preparation or local rotations.  

Future directions include (i) establish a formal classification of two-qubit three-parameter nonlocal unitaries to convert our numerical evidence into a definitive theorem, (ii) determine whether a single ancilla qubit and two-term Kraus operators always suffice for perfect-fidelity dynamics matching in arbitrary two-qubit experiments, (iii) extend the framework to higher-dimensional systems and multipartite environments, and (iv) analyze robustness to realistic imperfections (gate errors, finite sampling, and temporal drift), potentially by embedding the search for local controls in an error-aware optimization framework.

Overall, these results support a control-centric view of open-system dynamics: modest local interventions---measurements, basis changes, and minimal ancilla-assisted channels---can often recast correlated dynamics as if they arose from product-state inputs. Establishing the scope and robustness of these methods in scalable quantum platforms is an important next step.

\section*{Data Availability Statement}

The data and software that support the findings of this study are publicly available through the Dynamics Matching GitHub repository \cite{github_dilley} maintained by D. Dilley, A. Gonzales, J. Larson, and M. Byrd. All code and generated data necessary to reproduce the results presented in this work are available in the repository.

\section*{Acknowledgments}
This material is based upon work supported by the U.S.~Department of Energy, Office of Science, Accelerated Research in Quantum Computing, Fundamental Algorithmic Research toward Quantum Utility (FAR-Qu). This material is based upon work supported by the U.S. Department of Energy, Office Science, Advanced Scientific Computing Research (ASCR) program under contract number DE-AC02-06CH11357 as part of the InterQnet quantum networking project.


\bibliography{modeling}

\vfill
\framebox{\parbox{.90\linewidth}{\scriptsize The submitted manuscript has been
created by UChicago Argonne, LLC, Operator of Argonne National Laboratory
(``Argonne''). Argonne, a U.S.\ Department of Energy Office of Science
laboratory, is operated under Contract No.\ DE-AC02-06CH11357.  The U.S.\
Government retains for itself, and others acting on its behalf, a paid-up
nonexclusive, irrevocable worldwide license in said article to reproduce,
prepare derivative works, distribute copies to the public, and perform publicly
and display publicly, by or on behalf of the Government.  The Department of
Energy will provide public access to these results of federally sponsored
research in accordance with the DOE Public Access Plan
\url{http://energy.gov/downloads/doe-public-access-plan}.}}

\appendix

\section{Transformations to Correlations}\label{sec:TransfomToCorr}

The way correlations evolve between the system and environment depends on whether we apply a measurement or a unitary to the local system. For simplicity, we consider only the application of the measurement \(M_+\), which changes the local Bloch vectors and two-qubit correlation matrix as follows \cite{Gonzales_2024_Some}:
\begin{align} \label{Eq:Transformations}
    \vec{b}' &= \frac{\vec{b} + \epsilon \; \mathcal{T}_\rho^T \hat{n}}{1 + \epsilon \; \hat{n} \cdot \vec{a}}, \\ \nonumber
    \vec{a}' &= \frac{\sqrt{1-\epsilon^2} \; \vec{a} + (1-\sqrt{1-\epsilon^2}) (\hat{n} \hat{n}^T) \cdot \vec{a} + \epsilon \; \hat{n}}{1 + \epsilon \; \hat{n} \cdot \vec{a}}, \\ \nonumber
    \mathcal{T}_\rho' &=  \frac{\sqrt{1-\epsilon^2} \mathcal{T}_\rho + (1-\sqrt{1-\epsilon^2}) (\hat{n} \hat{n}^T) \cdot \mathcal{T}_\rho + \epsilon \; \hat{n} \vec{b}^T}{1 + \epsilon \; \hat{n} \cdot \vec{a}}.
\end{align}
Here, \(\mathcal{T}_\rho\) denotes the initial correlation matrix, and the probability of obtaining this result is given by $(1/2)(1 + \epsilon \; \hat{n} \cdot \vec{a})$. 
The probability of the ``$-$" outcome is given by flipping the sign on $\epsilon$, so we can easily determine the trade-off between the parameters and the likely outcome. In practice, we discard any event that is not associated with the chosen measurement output.

As a useful check, we see that the projective measurement ($\epsilon = 1$) completely destroys the correlations with the environment, leaving the overall state in product form; that is, $\mathcal{T}_\rho' = \vec{a}' \cdot (\vec{b}')^T$. However, this may come at the cost of fidelity, since we assume initial correlations between the system and environment, which restrict the local state to being mixed. For example, if the initial Bloch vector of the system is $\{0, 0, 0\}$, the resulting fidelity is only one-half in the worst-case scenario. This guarantees a product form in the end, but with minimal preservation of the state of the system, as we show in Section \ref{Sec:Fidelity}. In practice, we ideally want the state to be pure and completely separate from the environment. However, a projective measurement will also eliminate correlations with data qubits that are entangled with the system.


When we instead apply a local unitary $\mathcal{L}$, the effect on the correlations is much simpler and does not destroy correlations with other data qubits. The effect is merely a local change of basis. By taking $\mathcal{G}$ as a local unitary, we get the transformations
\begin{align} \label{Eq.Correlation_Transformations}
    \vec{b}' = \vec{b}, \quad \vec{a}' = \mathcal{O}_\mathcal{L} \cdot \vec{a}, \text{  and  } \mathcal{T}_\rho' = \mathcal{O}_\mathcal{L} \cdot \mathcal{T}_\rho.
\end{align}
Here, the operation \(\mathcal{O}_\mathcal{L}\) is the orthogonal matrix induced by \(\mathcal{L}\), with elements given by
\begin{align}
    \mathcal{O}_{ij} = \frac{1}{2} \operatorname{Tr}(\hat{\sigma}_i \mathcal{L} \hat{\sigma}_j \mathcal{L}^\dagger).
\end{align}
In this framework, there are two unitaries that map to the same orthogonal rotation on the original Bloch vector due to SU(2) being able to double-cover SO(3). The restriction of only being able to access the system prevents us from being able to diagonalize the correlation matrix, which would allow us to always $U$-generate the system dynamics with an initial product state. Next, we examine how these transformations affect fidelity.

\section{Fidelities of Different Strategies} \label{Sec:Fidelity}
Suppose we begin with an experiment described by the pair $(U, \rho^{SE})$, where the system dynamics cannot be $U$-generated by a product state input. By applying a measurement before the global unitary, we can try to enforce the product form. Note that if we find a solution for the measurement $M_-$, we know that there also exists a solution for $M_+$ and vice versa. So, the task is to find a measurement, $M_+$, and a valid environmental Bloch vector \(\vec{\zeta}^E\) such that the relation
\begin{align}\label{Eq:equality}
    &\tr_E \big[ U (M_+ \otimes \mathbb{I}) \rho^{SE} (M_+ \otimes \mathbb{I}) U^\dagger \big] \\ \nonumber
    = &\tr_E \big[ (U M_+ \rho^S M_+) \otimes \zeta^E U^\dagger \big]
\end{align}
is satisfied. To ensure that we preserve the initial state as much as we can, we optimize the fidelity between $\rho^S$ and $\rho^S_M = M_+ \rho^S M_+^\dagger$. This value is given by
\begin{align}
    F(\rho, \rho_\pm) = 1 - \frac{[1-(\vec{a}\cdot\hat{n})^2][1-\sqrt{1-\epsilon^2}]}{2(1\pm\epsilon \vec{a}\cdot\hat{n})},
\end{align}
where $\rho_{\pm} \defined M_{\pm}\rho^S M_{\pm}/\tr(M_{\pm} \rho^S M_{\pm})$. We obtained this exact solution by using the explicit single-qubit fidelity equation given by 
\begin{align}
    F(\rho, \sigma) = \tr (\rho \cdot \sigma) + 2 \sqrt{\text{det}(\rho) \cdot \text{det}(\sigma)}
\end{align}
and simplifying. Notice that the fidelity is affected by the magnitude of the local Bloch vector $\vec{a}$ and the direction $\hat{n}$. If its initial magnitude is 1, then $\hat{n} = \hat{a}$ for any $\epsilon$-value will result in a maximum fidelity of 1. Thus, $\epsilon$ should not be interpreted as a uniform measure of disturbance for all Bloch vectors $\hat{a}$; that is, it does not always represent how ``strong" the measurement is. On the other hand, if the local Bloch vector $\hat{a} = \{0,0,0\}$, then the fidelity monotonically decreases with an increase of $\epsilon$, which we use to optimize in some of our examples.

Now we consider the case when a local unitary transformation, $\mathcal{L}$, is used to $U$-generate system dynamics with an initial product state. The fidelity between the state of the system $\rho^S$ and its rotated version $\mathcal{L} \rho^S \mathcal{L}$ is given by
\begin{align}
    F_{\mathcal{L}} = 1 - \frac{1}{2} \left( \vec{a} - [\mathcal{O}_\mathcal{L} \cdot  \vec{a}] \right) \cdot \vec{a}.
\end{align}
This equation gives us insight into how fidelity changes, as it achieves maximum fidelity (\(F = 1\)) when $\vec{a}$ is left invariant by $\mathcal{O}_L$ and minimum fidelity (\(F = 1 - \vec{a} \cdot \vec{a}\)) when it rotates from $\vec{a}$ to $-\vec{a}$. Thus, we would like to rotate about the Bloch vector $\vec{a}$ whenever possible to perfectly preserve the fidelity of the initial system while ensuring CP dynamics. This is possible in many experiments, as we show in Appendix  \ref{Sec:Numerical_Local_Operations}, since the local unitary on the system also induces the same orthogonal rotation to the left side of the correlation matrix $\mathcal{T}_\rho$.

\section{Repeated Measurements}\label{Sec:Repeated_Measurements}

A natural question is whether a sequence of repeated measurements during preprocessing is equivalent to a single measurement with a different value of the parameter $\epsilon$. If this were true, then there would be no advantage to using repeated measurements in manipulating system dynamics. However, we show that this is generally not the case when we restrict ourselves to measuring in the same direction with the same value of $\epsilon$ each round. As discussed previously, the measurement in Eq.~\eqref{Eq:measurement} is related to the weak measurement ${A_{0|1}}$ in \cite{Brun_2002} by a unitary rotation. Our rotated measurements are weak in the sense that we may perturb the initial state of the system minimally while ensuring the dynamics-matching condition is satisfied. This may mean that near-projective measurements are necessary if the initial state is close to pure and has very little correlation with the environment. In this sense, the measurement is weak to us but strong in other contexts.

To show that these repeated measurements are inequivalent to a single round of measurements with a different value of $\epsilon$, we examine the effect of multiple rounds of measurements on an arbitrary initial state $\ket{\psi}$. First, we observe that $M_+ M_- \propto \mathbb{I}$, so that \begin{align} M_+^m M_-^n = \begin{cases} M_+^{m-n} \text{ if } m > n, \ M_-^{n-m} \text{ if } n > m, \end{cases} \end{align} up to a factor that does not affect the state of the system. The reason is that each measurement operator acts as the inverse of the other after renormalization.

Next, we consider both measurement operators raised to the $k$th power. We find that
\begin{align} M_\pm^k = \dfrac{1}{2} \left\{ \tilde{\epsilon}_+^{k} \mathbb{I} \pm \tilde{\epsilon}_-^{k} \hat{n} \cdot \vec{\sigma} \right\},
\end{align}
where the new coefficients are given by
\begin{align} \label{Eq:New_Coefficients}
\tilde{\varepsilon}_\pm^{k} = \left( \dfrac{1 + \epsilon'}{2} \right)^{k/2} \pm \left( \dfrac{1 - \epsilon'}{2} \right)^{k/2}.
\end{align}
These coefficients resemble the original ones in Eq.~\eqref{Eq:Measurement_Coefficients}, and it will become clear that whenever $\tilde{\varepsilon}_+^k = \varepsilon_+$ for some $k$, we cannot have $\tilde{\varepsilon}_-^k = \varepsilon_-$. The reason we include the primes in Eq.~\eqref{Eq:New_Coefficients} is that the repeated measurements will satisfy $\epsilon' \leq \epsilon$, with equality  occurring only  when $k = 1$. The exact solution for when $\tilde{\epsilon}_{+}^k = \epsilon_+$ is given by $\epsilon = \tilde{\epsilon}_{+}^k \sqrt{2 - (\tilde{\epsilon}_{+}^k)^2}$. On the other hand, the exact solution for when $\tilde{\epsilon}_{-}^k = \epsilon_-$ is given by $\epsilon = \tilde{\epsilon}_{-}^k \sqrt{2 - (\tilde{\epsilon}_{-}^k)^2}$. Define the function $g(q) = q \sqrt{2-q^2}$. Note that $g(\sqrt{2-q^2}) = g(q)$. Hence, $g(\tilde{\epsilon}_{+}^k) = g(\tilde{\epsilon}_{-}^k)$ \textit{iff} $\sqrt{2-(\tilde{\epsilon}_{+}^k)^2} = \tilde{\epsilon}_{-}^k$. So, if $\tilde{\epsilon}_{+}^k \sqrt{2 - (\tilde{\epsilon}_{+}^k)^2} = \tilde{\epsilon}_{-}^k \sqrt{2 - (\tilde{\epsilon}_{-}^k)^2}$, we must satisfy $(\tilde{\epsilon}_{+}^k)^2 + (\tilde{\epsilon}_{-}^k)^2 = 2$. 

From the previous equality, we calculate an equivalent expression given by
\begin{align} \label{Eq.Equivalent_Equality}
   h(\epsilon') \defined (1+\epsilon')^k + (1-\epsilon')^k = 2^k.
\end{align}
The derivative is given by
\begin{align}
    \dfrac{\partial h(\epsilon')}{\partial \epsilon'} = k[(1+\epsilon')^{k-1} - (1-\epsilon')^{k-1}],
\end{align}
which is strictly positive for $k \geq 2$ and $0 < \epsilon' <1$. Therefore, repeated measurements always cause the function $h(\epsilon')$ to increase for all $\epsilon'$-values. We also calculate $h$ at the end points to obtain $h(0) = 2$ and $h(1) = 2^k$, which means that $2 = h(0) < h(\epsilon') < h(1) = 2^k$ for $0 < \epsilon' < 1$. For the remaining case $k=1$, we have $h(\epsilon') = 2$ for all $\epsilon'$, and Eq.~\eqref{Eq.Equivalent_Equality} is satisfied trivially. 
Therefore, repeated measurements cannot be simulated by a single measurement unless they represent a series of the same projective measurement. For the purposes of this paper, we restrict our analysis to single measurements.

\section{Optimal Measurement Can Be Projective}
\label{Sec:Full_Projection_Example}
 When optimizing the fidelity in order to reproduce the system dynamics from a product state, one may ask whether it is necessary to perform a projective measurement that completely destroys all correlations between the system and environment. This would represent one of the worst-case scenarios if the system is entangled with other data qubits, and we show that it can happen with the next example. Consider when the initial correlations with the environment are given by the Bell state $\rho^{SE} = \op{\Phi^+}{\Phi^+}$. Let the global evolution be $U = SWAP \circ CNOT$. Then the system evolution is given by
 \begin{align}
     \rho^S = \dfrac{1}{2} \mathbb{I} \quad \overset{U}{\rightarrow} \quad \rho^S_U = \op{0}{0},
 \end{align}
which cannot be $U$-generated by a product state.

If a local measurement is applied before $U$, then the solution for the environmental Bloch vector is given by
\begin{align}
    &\zeta_x = \epsilon n_x, \; \zeta_y = \dfrac{n_y(\sqrt{1-\epsilon^2}-1)}{\epsilon}, \\
    &\text{ and } \zeta_z = \dfrac{n_z^2+(1-n_z^2)\sqrt{1-\epsilon^2}}{\epsilon n_z}.
\end{align}
This ensures that we can generate the same system dynamics with an initial product state. Notice that the solution is  valid only when $\epsilon = 1$ since we must satisfy the normalization condition for $\hat{n} \defined \{n_x,n_y,n_z\}$. The solution for the environmental Bloch vector then becomes $\vec{\zeta} = \{n_x, -n_y, n_z\}$. 
This example illustrates a fundamental limitation of weak measurements: in certain settings, a partial disturbance of the system--environment correlations is insufficient. A full projective measurement is required to make the product-state description possible.

\section{Werner States with Local Measurements}
\label{sec:Werner}
To broaden the discussion, we now examine a class of mixed states that is $U \otimes U$-invariant, called Werner states. The Werner state $W(\lambda)$ is defined as a convex combination of the antisymmetric Bell state $\op{\Psi^-}{\Psi^-}$ and the maximally mixed state on two qubits. Explicitly,
\begin{align} \label{Eq:Werner_Family}
    W(\lambda) \defined \lambda \op{\Psi^-}{\Psi^-} + (1-\lambda) \dfrac{\mathbb{I} \otimes \mathbb{I}}{4},
\end{align}
where the parameter $\lambda$ lies between 0 and 1. Importantly, Werner states have vanishing local Bloch vectors, since the local states are maximally mixed. If we assume that the global evolution is given by the CNOT gate, then the local measurement must have $n_y = 0$ for the second component of the local Bloch vectors to be equal; otherwise, we can never satisfy the dynamics-matching condition. If the environmental part of the product state has a Bloch vector given by $\vec{\zeta} \defined \{\zeta_x, \zeta_y, \zeta_z\}$, then the only condition for equality is
\begin{align}
\label{Eq:rSolution}
    \zeta_x = -\left[ \sqrt{1-\epsilon^2} \cdot \dfrac{1}{n_x} + (1-\sqrt{1-\epsilon^2}) \cdot n_x \right] \cdot \dfrac{\lambda}{\epsilon},
\end{align}
where we need the optimal $\hat{n}_x$-value that allows us to minimize $\epsilon$ for a valid Bloch vector $\vec{\zeta}$ constrained to $\zeta_x \in [-1, 1].$ This will ensure that we are also maximizing the fidelity.

Since the parameter $\epsilon \in [0,1]$, the term $-\sqrt{1-\epsilon^2}$ is less than or equal to $-(1-\sqrt{1-\epsilon^2})$ for $\epsilon \in [0,\sqrt{3}/2]$, which divides our optimization into two parts. The minimum $\epsilon$ must occur when $n_x = 1$ to ensure there is a valid environmental Bloch vector on the first interval. The reason is that the $1/n_x$ term needs to be at its smallest for $\epsilon \in [0,\sqrt{3}/2]$ to prevent $\zeta_x$ from becoming less than $-1$ as $\epsilon$ becomes smaller. Since $n_x \in [-1,1]$ and the absolute value of the left-hand term is always larger than or equal to the absolute value of the right-hand term for all valid $n_x$-values, then $\zeta_x$ must be determined at the end points for $n_x$; let us say 1 since -1 gives us the same bound by symmetry. This enforces $\zeta_x = -\lambda/\epsilon$, which tells us that $\epsilon_{min} = \lambda$ since $\epsilon \geq 0$. In other words, when $\lambda$ is between the values $0$ and $\sqrt{3}/2$, the minimum $\epsilon$ needed to $U$-generate the system dynamics with a product state is simply equal to $\lambda$. Thus, $\epsilon_{min} = \lambda$ for $\lambda \in [0,\sqrt{3}/2]$.
\newline
\newline
\indent We now consider the interval in which $\epsilon \in [\sqrt{3}/2,1]$ so that the term $-(1-\sqrt{1-\epsilon^2})$ is less than or equal to $-\sqrt{1-\epsilon^2}$. In this case, we would like $n_x$ to be in the interval $[-1,0]$ since we cannot simply minimize the right term in Eq.~(\ref{Eq:rSolution}) because the terms on the left blow up to infinity. When our solution for $\zeta_x$ is unit ($\zeta_x \rightarrow \pm \infty$ in the limit as $\epsilon\rightarrow 0$, then by continuity $\zeta_x$ must be one or minus one when $\epsilon$ is minimum), we get the equality
\begin{align}
    \epsilon = -\dfrac{n_x^3 - (1-n_x^2)\sqrt{\lambda^2 + n_x^2(1-2\lambda^2)}}{n_x^2+(1-n_x^2)^2\lambda^2} \cdot \lambda.
\end{align}
Taking the derivative of $\epsilon$ with respect to $n_x$ and setting it equal to zero, and applying the boundary conditions on $n_x$ and $\epsilon$, we get the solutions
\begin{align}
    n_x = -\dfrac{1}{\sqrt{4 \lambda^2-2}} \; \text{and} \; n_z = \pm \sqrt{1-\dfrac{1}{\sqrt{4 \lambda^2-2}}}.
\end{align}

Plugging $n_x$ into our definition of $\epsilon$ yields the final solution of
\begin{align}
    \epsilon_{min} = \begin{cases}
    \lambda \; \text{ if } \lambda \in [0, \sqrt{3}/2] \\
    \dfrac{2\lambda \sqrt{4\lambda^2-2}}{4\lambda^2-1} \; \text{ if } \lambda \in [\sqrt{3}/2, 1]
    \end{cases}
\end{align}
for the minimum $\epsilon$ (maximum fidelity) that ensures the system dynamics is $U$-generated by a product state. As expected, the minimum value of $\epsilon$ depends on the parameter $\lambda$ that defines the Werner state. In Fig.~ \ref{Fig:Werner_Fidelity_Plot} we plot the optimal fidelities for the entire family of Werner states given in Eq.~\eqref{Eq:Werner_Family}. The plot monotonically decreases with an increase in $\lambda$, which means that it is harder to preserve fidelity with Werner states nearest to the singlet state. For the maximally mixed state, $\lambda = 0$, we can do this without any measurement, since it is already in the product form.
\begin{figure}[h!]
    \centering
    \includegraphics[scale = 0.67]{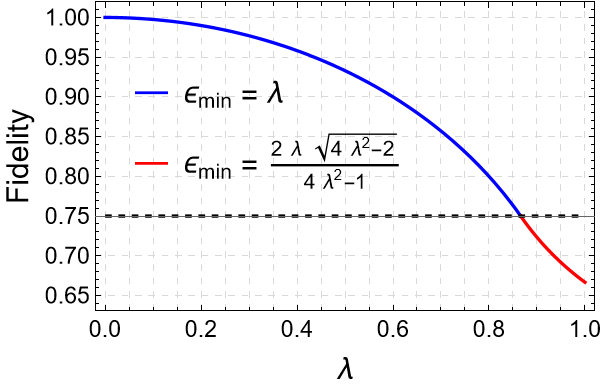}
    \caption{Optimal fidelity for $U$-generating the system dynamics by a product state over all family of Werner states $W(\lambda)$.}
    \label{Fig:Werner_Fidelity_Plot}
\end{figure}

\section{Proof of Theorem \ref{Thm:Two_Parameter_Theorem}}
\label{Supp:Theorem_1}

We want the following equality to hold for some local unitary $V$, the system, and  some valid environmental state $\zeta^E$:
    \begin{align} \label{Eq:Equality_Proof}
        \tr_E(U V^S \rho^{SE} V^{S^\dagger} U^\dagger) = \tr_E(U V^S \rho^S V^{S^\dagger} \otimes \zeta^E U^\dagger).
    \end{align}
    We can define the global evolution as $U = (\mathcal{L}_1 \otimes \mathcal{L}_2) \Omega (\mathcal{R}_1 \otimes \mathcal{R}_2)$ \cite{Khaneja_2001, Kraus_2001, Zhou_2007} and eliminate the local unitaries on the left, $\mathcal{L}_1$ and $\mathcal{L}_2$, since they make no difference to the equality in Eq.~\eqref{Eq:Equality_Proof}. We also observe that $R_1 V^S = R_1 V^S R_1^\dagger R_1 \rightarrow V^S R_1^\dagger$, since $V^S$ is chosen over all local unitaries. The same is true for the local environmental state $\zeta^E$, where $R_2 \zeta^E R_2^\dagger \rightarrow \zeta^E$ can be any state. This means that if a solution exists for the equality
    \begin{align} \label{Eq:Condition_Reduced}
        &tr_E[\Omega V^S (\mathcal{R}_1 \otimes \mathcal{R}_2) \rho^{SE} (\mathcal{R}_1^\dagger \otimes \mathcal{R}_2^\dagger) V^{S^\dagger} \Omega^\dagger] \\
        = &\tr_E[\Omega (V^S \mathcal{R}_1 \rho^S \mathcal{R}_1^\dagger V^{S^\dagger}) \otimes (\mathcal{R}_2 \zeta^E \mathcal{R}_2^\dagger) \Omega^\dagger],
    \end{align}
    then it also exists for Eq.~\eqref{Eq:Equality_Proof}.
    
    Since we want to prove Theorem \ref{Thm:Two_Parameter_Theorem} for all two-qubit states, equality must hold over any local unitaries $\mathcal{R}_1$ and $\mathcal{R}_2$. This will result in a generic state that is captured by any $\rho^{SE}$, so we can throw away these right local unitaries and only ensure that
    \begin{align} \label{Eq:Reduced_Equality_Equivalence}
        \tr_E(\Omega V^S \rho^{SE} V^{S^\dagger} \Omega^\dagger) = \tr_E(\Omega V^S \rho^S V^{S^\dagger} \otimes \zeta^E \Omega^\dagger)
    \end{align}
can hold for any two-parameter family of $\Omega$ and any $\rho^{SE}$, for some local unitary $V$ acting on the system, and some environmental state $\zeta^E$. When we get a solution for Eq.~\eqref{Eq:Reduced_Equality_Equivalence}, we can also obtain a solution for Eq.~\eqref{Eq:Equality_Proof} by conjugating $V^S$ by $R_1^\dagger$ and $\zeta^E$ by $R_2^\dagger$. Using the reduced form in Eq.~\eqref{Eq:Reduced_Equality_Equivalence}, we must show that the three equalities
\begin{widetext}
    \begin{align} \label{Eq:Condition_1}
(t_{32} - a_3 \zeta_2) \cos(2\alpha_3) \sin(2\alpha_2) + [(-t_{23} + a_2 \zeta_3) \cos(2\alpha_2) + (b_1 - \zeta_1) \sin(2\alpha_2)] \sin(2\alpha_3) &= 0, \\ \label{Eq:Condition_2}
(-t_{31} + a_3 \zeta_1) \cos(2\alpha_3) \sin(2\alpha_1) + [(t_{13} - a_1 \zeta_3) \cos(2\alpha_1) + (b_2 - \zeta_2) \sin(2\alpha_1)] \sin(2\alpha_3) &= 0, \\ \label{Eq:Condition_3}
\text{and } (t_{21} - a_2 \zeta_1) \cos(2\alpha_2) \sin(2\alpha_1) + [(-t_{12} + a_1 \zeta_2) \cos(2\alpha_1) + (b_3 - \zeta_3) \sin(2\alpha_1)] \sin(2\alpha_2) &= 0
    \end{align}
\end{widetext}
can always be satisfied for some local unitary inserted between $U$ and $\rho^{SE}$ ($\rho^S \otimes \zeta^E$) on system $S$. We derived these equations by solving Eq.~\eqref{Eq:Reduced_Equality_Equivalence} exactly. The vectors $\vec{a} \defined \{ a_1, a_2, a_3 \}$, $\vec{b} \defined \{ b_1, b_2, b_3 \}$, and $\vec{\zeta} \defined \{ \zeta_1, \zeta_2, \zeta_3 \}$ are the local Bloch vectors for $V^S \rho^S V^{S^\dagger}$, $\rho^E$, and $\zeta^E$, respectively. The correlation matrix of $V^S \rho^{SE} V^{S^\dagger}$ is given by $\{t_{ij}\}_{i,j = 1}^3$. 

In this section we will review two cases: the one-parameter family and the two-parameter family of global unitaries. In the one-parameter family, we have that $\alpha_i \notin \mathcal{S}$ for only one of $i \in\{ 1, 2, 3 \}$. Let us assume, without loss of generality, that $\alpha_2 \in \mathcal{S}$ and $\alpha_3 \in \mathcal{S}$. Then we must satisfy the relations
\begin{align}
    (a_3 \zeta_1 - t_{31}) \sin(2\alpha_1) = (a_2 \zeta_1 - t_{21}) \sin(2\alpha_1) = 0
\end{align}
for all $\alpha$ and some choice of $\zeta_1$. We have the freedom to rotate the correlation matrix and the local Bloch vector $\vec{a}$. We can simply rotate the correlation matrix on the left side using Givens rotations to make $t_{21} = 0$ and $t_{31} = 0$, and then set $\sigma_1 = 0$ to satisfy these equalities for some local unitary acting on the system. This is always possible because we can apply whatever Givens rotations we want on the left of the correlation matrix by acting on the system with the appropriate local unitary. The same argument can also be made in the case when $\alpha_1 \in \mathcal{S}$ and $\alpha_2 \in \mathcal{S}$ or $\alpha_1 \in \mathcal{S}$ and $\alpha_3 \in \mathcal{S}$.

We now look at what occurs with the two-parameter family when only $\alpha_3 \in \mathcal{S}$ and $\zeta_3 = b_3$. By symmetry, we can come to the same conclusion for the other two-parameter families. Therefore, it suffices to show only what happens in this case. We obtain the conditions
\begin{align}
    (t_{32} - a_3 \zeta_2) \sin(2\alpha_2) &= 0, \\
    (-t_{31} + a_3 \zeta_1) \sin(2\alpha_1) &= 0, \\
    \notag\text{and } \; \; \; \; \; (t_{21} - a_2 \zeta_1) \cos(2\alpha_2) \sin(2\alpha_1)&\\ 
    + (-t_{12} + a_1 \zeta_2) \cos(2\alpha_1) \sin(2\alpha_2) &= 0,
\end{align}
in which we set $\zeta_1 = 0$ and $\zeta_2 = 0$. If we apply the Givens rotations
\begin{equation}
\left(
    \begin{array}{ccc}
        0 & 0 & 1 \\
        \cos(x) & -\sin(x) & 0 \\
        \sin(x) & \cos(x) & 0
    \end{array}
    \right) \text{ and }
\left(
    \begin{array}{ccc}
        \cos(y) & 0 & -\sin(y)  \\
        0 & 1 & 0 \\\
        \sin(y) & 0 & \cos(y),
    \end{array}
    \right)
\end{equation}
we can solve for $x$ and $y$ to always make $t_{32} = 0$ and $t_{31} = 0$. This covers the first and second equalities, so that we now only have to zero out the third component given by
\begin{align} \label{Eq:Final_Component}
    t_{21} \cos(2\alpha_2) \sin(2\alpha_1) -t_{12} \cos(2\alpha_1) \sin(2\alpha_2).
\end{align}
To do this, we apply a final Givens rotation given by
\begin{equation}
    \left(
    \begin{array}{ccc}
        \cos(z) & -\sin(z) & 0 \\
        \sin(z) & \cos(z) & 0 \\
        0 & 0 & 1
    \end{array}
    \right)
\end{equation}
and solve for the variable $z$ that makes the expression in Eq.~\eqref{Eq:Final_Component} equal to 0. We get the transformations $t_{21} \rightarrow t_{21} \cos(z) + t_{11} \sin(z)$ and $t_{12} \rightarrow t_{12}\cos(z) - t_{22} \sin(z) $ for the current correlation terms after the first two Givens rotations. We can always find some $z$-value that makes the value in Eq.~\eqref{Eq:Final_Component} zero while leaving the first two equalities the same.

\section{Proof of Theorem \ref{Thm:Diagonal_States}}
\label{Supp:Theorem_2}

We begin the proof by noting that the conditions in Eqs. (\ref{Eq:Condition_1})--(\ref{Eq:Condition_3}) can always be satisfied whenever one of the angles has the form $k \pi$ for $k \in \mathbb{W}$, while the other angles can be anything. For example, the solution for $\alpha_i = 0$ is given by $\zeta_j = 0$ when $i \neq j$ and by $\zeta_j = b_j$ when $i = j$. Now we assume that none of the angles are equal to $k \pi$ for the next part of the proof to complete generality. The conditions can now be made into a compact expression given by
    \begin{align}
        (S+M) \cdot \vec{\zeta} = S \cdot \vec{b}
    \end{align}
    for the matrices $S = \text{diag}(s_2 s_3, s_1 s_3, s_2 s_3)$ and
    \begin{align}
        M \defined \left(\begin{array}{ccc}
            0 & a_3 c_3 s_2 & -a_2 c_2 s_3 \\
            -a_3 c_3 s_1 & 0 & a_1 c_1 s_3 \\
            a_2 c_2 s_1 & -a_1 c_1 s_2 & 0
        \end{array}\right)
    \end{align}
    if we define $s_i \defined \text{sin}(2 \alpha_i)$ and $c_i \defined \text{cos}(2 \alpha_i)$. We can further reduce the expression to
    \begin{align}
        (\mathbb{I} + S^{-1} \cdot M) \cdot \vec{\zeta} = \vec{b}.
    \end{align}
    if we use the fact that no $\alpha_i = k \pi$ to ensure that $S$ is full rank. After simplifying, we calculate the magnitude for the solution to be
    \begin{align}
        \zeta^2 = \dfrac{b^2 + (\vec{b} \cdot \vec{a_t})^2}{1 + \vec{a_t} \cdot \vec{a_t}} = \dfrac{1+\text{cos}(\kappa)a_t^2}{1 + a_t^2} b^2
    \end{align}
    for $\vec{a_t} = \{ a_1 \text{cot}(2\alpha_1), a_2 \text{cot}(2\alpha_2), a_3 \text{cot}(2\alpha_3) \}$. We need this magnitude to always be less than or equal to 1 to ensure that a solution exists for $\vec{\zeta}$. We see that
    \begin{align}
        \dfrac{1+\text{cos}(\kappa)a_t^2}{1 + a_t^2} b^2 \Rightarrow b^2 \leq \dfrac{1 + a_t^2}{1+\text{cos}(\kappa)a_t^2} \in [1, 1 + a_t^2].
    \end{align}
    This will always be true for any choice of non-zero angles and set of local Bloch vectors.

\section{Proof of Theorem \ref{Thm:Non-Entangling_Unitary}}
\label{Supp:Theorem_3}

The condition in Eq.~\eqref{Eq:Non_Entangling_Condition} is the explicit form for the normalized entangling power \cite{Zanardi2000, Zanardi2001} of a global two-qubit unitary with the nonlocal portion given in Eq.~\eqref{Eq:Nonlocal_Part}. It is defined as the average entanglement generated by the global unitary $U$ when acting on the set of all two-qubit pure states. For all two-qubit unitaries, the entangling power is also equal to the disentangling power. Thus, when the entangling power vanishes, the unitary has no capacity to change the entanglement of $\rho^{SE}$. If the entanglement does not change, then the dynamics on the system can always be $U$-generated by a product state. One can easily see this because the only nonlocal two-qubit gates that are not entangling are locally equivalent to SWAP \cite{Makhlin2002}.

\section{Numerical Enforcement of Dynamics Matching}
\label{Sec:Numerical_Local_Operations}

We generated 1,000 two-qubit states and global unitaries that satisfy a number of conditions to increase our chances of obtaining cases that do not satisfy Eq.~(\ref{Eq:equality}). The constraints are given by
\begin{align}
    t_{ij} \rightarrow \text{sgn}[\text{cos}(2\alpha_i) \text{sin}(2\alpha_j)] |t_{ij}|,
\end{align}
with $|t_{ij}| \in[1/4, 1]$ and $a_i, b_i \in [-1/2, 1/2]$ with
\begin{align}
    b_k \rightarrow \dfrac{1}{2} \sum_{i,j} |\epsilon_{ijk}| \; \text{sgn}[\text{sin}(2\alpha_i) \text{sin}(2\alpha_j)] |b_k|
\end{align}
for the Levi--Cevita tensor $\epsilon_{ijk}$. From Eqs.~(\ref{Eq:Condition_1})--(\ref{Eq:Condition_2}) we see that these additional conditions will make it ``harder" for the left-hand sides to equal 0. Therefore, we increase the probability that we will run into instances of Eq.~(\ref{eq:initialCond0}) not holding for a valid $\zeta^E$. We only restrict $\alpha_i$ to be on the interval $ [0, 2\pi]$ and let $t_{ii} = 0$ for all $i$ since they do not appear in Eqs.~(\ref{Eq:Condition_1})--(\ref{Eq:Condition_2}). Since we would like to increase the correlation strength among other parameters $t_{ij}$ for when $i \neq j$, we must minimize $|t_{ii}|$ over each $i$. When the Bloch vector of $\zeta^E$ has a magnitude greater than 1 in order for Eq.~(\ref{Eq:equality}) to be true, we hold onto it for further processing. Otherwise, we throw it away since the condition already holds trivially. We retained 402 of the cases, which ranged in magnitudes, $|\vec{\zeta}|$, of slightly over 1 to more than 20 in some cases.

\begin{figure}
    \centering
    \includegraphics[scale = 0.40]{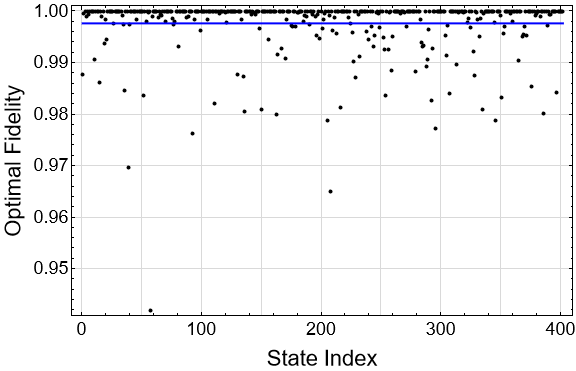}
    \caption{Plot of the optimal fidelity for each of the 402 experiments $(U_i, \rho^{SE}_i)$. In each case, a local unitary $V$ acts on the system such that the dynamics generated by $U_i$ can be reproduced by using a product state, and the resulting fidelity between $\rho^S$ and $V \rho^S V^\dagger$ is shown. The blue line indicates an average fidelity of roughly 99.8\%.}
    \label{Fig:Fidelity_vs_State}
\end{figure}

The goal is to determine whether there exists a case in which no local unitary acting on the system allows us to $U$-generate the system dynamics using a product state. Where such a unitary exists, we optimize the fidelity between the reduced system state before and after the unitary is applied. As shown in Fig.~\ref{Fig:Fidelity_vs_State}, we observe that not only is it always possible to satisfy the dynamics-matching condition in Eq.~(\ref{Eq:equality}) using a single local unitary, but in most cases this can be achieved with unit fidelity. Operationally, this corresponds to rotating the local system state about its own Bloch vector, thereby modifying only the nonlocal correlations between the system and environment while leaving the local state invariant.

These numerical observations motivate the following conjecture.
\begin{conjecture}
There exists a local unitary transformation in the preprocessing step that allows the system dynamics of any two-qubit state to be generated from an initial product state undergoing the global evolution $U$.
\end{conjecture}
\noindent This would imply that a simple change in basis on the local system is sufficient to recover the product-state prescription. The initially correlated states and the nonlocal parameters for the global evolutions, given in Eq.~(\ref{Eq:Nonlocal_Part}), can be found on GitHub \cite{github_dilley}. We also provide the Mathematica notebooks used to perform the numerical and symbolic calculations throughout this work.

\end{document}